\documentclass[12pt]{article}
\usepackage{amsmath, amsthm, amscd, amsfonts, amssymb, graphicx, color}
\usepackage[bookmarksnumbered, colorlinks]{hyperref}
\usepackage{float}
\usepackage{lipsum}
\usepackage{afterpage}
\usepackage[labelfont=bf]{caption}
\usepackage[nottoc,notlof,notlot]{tocbibind}
\usepackage{fancyhdr}
\usepackage{array,tabularx,booktabs}
\newcolumntype{Y}{>{\raggedright\arraybackslash}X}
\newcommand{\keyword}[1]{\par\noindent \textbf{Keywords:} #1 }

\newtheorem{defi}{Definition}[section]
\newtheorem{lemma}[defi]{Lemma}
\newtheorem{proposition}[defi]{Proposition}
\newtheorem{theorem}[defi]{Theorem}
\newtheorem{definition}[defi]{Definition}
\newtheorem{assumption}{Assumption}
\newtheorem{corollary}[defi]{Corollary}

\title{Renormalized Energy and Vortex Interaction in Finsler Ginzburg--Landau Models}

\author
{Y. Alipour Fakhri\thanks{Corresponding Author:
Faculty of Basic Sciences,
Department of Mathematics, Payame Noor University, Tehran, Iran.  E-mail: y\_ alipour@pnu.ac.ir}
}

\begin{document}
\maketitle

\begin{abstract}
We develop a Finsler Ginzburg--Landau framework for the analysis of vortex interactions in anisotropic superconductors.  
Within this setting, the Finsler structure encodes directional dependence of the condensate energy, yielding a renormalized energy $W_F$ that governs both equilibrium and dynamics of vortices.  
We derive the $\Gamma$--limit, establish the analytical structure and stability of $W_F$, and show that vortex motion follows a Finsler gradient flow exhibiting anisotropic dissipation and drift.  
This approach provides a unified geometric and physical model for anisotropic superconductivity.
\end{abstract}

%==============================================
\keyword{ Finsler geometry; Ginzburg--Landau vortices; anisotropic superconductivity; renormalized energy; gradient flow; vortex stability.}
 \\*
 \textbf{[2020] Mathematics Subject Classification}:53C60, 35Q56, 35J50, 82D55, 35B40.

%==================================
%=======================sec:Introduction==
\section{Introduction}\label{sec:Introduction}
%==========================================
The Ginzburg--Landau theory provides a fundamental framework for describing superconductivity through a complex order parameter $\psi$ and a gauge field $A$. In its classical form, this theory successfully captures the formation of quantized vortices, magnetic flux quantization, and the emergence of vortex lattices in type-II superconductors. The mathematical structure of the Ginzburg--Landau model has been extensively investigated from variational and analytical perspectives, leading to the development of the renormalized energy formalism, which describes the effective interaction among vortices in the limit of vanishing coherence length $\varepsilon \to 0$ \cite{Serfaty2014,IgnatJerrard2020}. In this asymptotic regime, the energy functional reduces to a geometric interaction law among discrete vortex centers, and its minimizers determine the equilibrium arrangements of vortices.

Despite the success of this approach in isotropic systems, it is insufficient to capture the anisotropic behavior exhibited by many real superconductors. In crystalline or layered materials, physical properties such as coherence length, magnetic penetration depth, and critical current density depend on direction, leading to anisotropic vortex motion and distorted lattice patterns. These effects cannot be described within the Euclidean framework underlying the standard Ginzburg--Landau model, which assumes isotropic quadratic dependence on the gradient of $\psi$. To overcome this limitation, several geometric generalizations have been proposed, replacing the Euclidean structure by a more general metric or tensorial dependence. However, these formulations remain essentially Riemannian, and therefore unable to capture directional dependence beyond quadratic symmetry.

A natural and rigorous way to model anisotropy in superconductors is provided by Finsler geometry. In this framework, the norm of a tangent vector $v\in T_xM$ is given by a position-dependent, non-quadratic function $F(x,v)$ that encodes directional information. Finsler geometry extends Riemannian geometry by allowing distinct local metric structures in different directions, making it ideally suited for describing anisotropic media and non-reciprocal transport phenomena \cite{BaoChernShen2000}. Within this setting, the differential operators, measure structures, and variational principles acquire direction-dependent weights, naturally reproducing anisotropic effects at the level of the Ginzburg--Landau equations.

In the companion work \cite{AlipourFakhri2025}, a Finsler Ginzburg--Landau model was introduced as a variational functional of the form
\begin{align*}
G_F[\psi,A] = \int_M \left( \frac{1}{2}\|D_A\psi\|_{F^*}^2 + \frac{1}{2\lambda}\|dA\|_{\gamma^*}^2 + \frac{1}{4\varepsilon^2}(1-|\psi|^2)^2 \right)d\mu_F,
\end{align*}
defined on a compact Finsler manifold $(M,F)$ with measure $d\mu_F$. The co-metric $F^*$ replaces the Euclidean norm in the kinetic term, yielding an anisotropic energy landscape governed by the geometry of the unit co-ball $\{\xi:F^*(x,\xi)\le 1\}$. Minimizers of this energy exhibit quantized vortices whose cores are determined by the zeros of $\psi$, while the gauge field $A$ mediates long-range magnetic interaction among them. The $\Gamma$--limit of this functional as $\varepsilon\to 0$ produces a geometric measure-theoretic description of vorticity supported on $(n-2)$--dimensional submanifolds \cite{IgnatJerrard2020}, and defines an anisotropic limiting energy that depends explicitly on the Finsler structure.

The present paper develops the analytical and physical theory of renormalized energy and vortex interaction in the Finsler Ginzburg--Landau framework.  
The analysis begins in Section~\ref{sec:Geometric Framework}, where the geometric structure of the model and the definition of the renormalized energy $W_F$ are introduced through the asymptotic expansion of the Finslerian functional.  
Section~\ref{sec:AnalyticalStructure} establishes the analytic properties of $W_F$, deriving its differential structure, Finsler gradient, and equilibrium conditions.  
In Section~\ref{sec:PhysicalInterpretation}, the physical meaning of these quantities is examined, revealing how anisotropy influences vortex alignment and the formation of directionally biased lattice structures in superconductors.  
Section~\ref{sec:Stability Analysis} develops the second-order structure of $W_F$, linking the Finslerian Hessian to the effective elastic response of the condensate and to the stability of vortex configurations.  
Finally, in Section~\ref{sec:AsymptoticDynamics}, the asymptotic dynamics of vortices are derived as a Finsler gradient flow, providing a physical picture of energy dissipation and anisotropic drift in the vortex system.

The main objective of this work is to demonstrate that the Finsler Ginzburg--Landau model naturally unifies the geometric and physical descriptions of anisotropic superconductivity.  
Within this setting, the anisotropy of the superconducting medium is represented directly by the Finsler structure, without resorting to ad hoc anisotropic coefficients.  
The renormalized energy $W_F$ thus serves both as a geometric invariant associated with the Finsler metric and as a physical potential governing the equilibrium and dynamics of vortices.  
This dual interpretation establishes a rigorous link between the geometry of Finsler spaces and the phenomenology of anisotropic superconductors, offering a new mathematical framework for studying directionally dependent superconducting phenomena from first principles.

%==========sec:RenormalizedEnergy========
\section{Renormalized Energy Formulation}
\label{sec:Geometric Framework}
%=========================================

The purpose of this section is to establish the analytic and geometric setting in which the renormalized energy for vortices in Finsler Ginzburg--Landau models is defined. We begin by recalling the asymptotic regime of the Finslerian Ginzburg--Landau functional introduced in \cite{AlipourFakhri2025}, together with the associated $\Gamma$--limit, and then proceed to derive the higher-order expansion that leads to the definition of the renormalized energy $W_F$.

Let $(M,F)$ be a compact, oriented, smooth Finsler manifold of dimension $n \geq 2$, endowed with a smooth Finsler measure $d\mu_F$ (either Busemann--Hausdorff or Holmes--Thompson) and a Riemannian co-metric $\gamma^*$ uniformly equivalent to $F^*$. For a complex scalar field $\psi : M \to \mathbb{C}$ and a real one-form $A \in \Omega^1(M)$, we recall the Finslerian Ginzburg--Landau functional
\begin{align}
G_F[\psi,A] = \int_M \left( \frac{1}{2}\|D_A\psi\|_{F^*}^2 + \frac{1}{2\lambda}\|dA\|_{\gamma^*}^2 + \frac{1}{4\varepsilon^2}(1-|\psi|^2)^2 \right) d\mu_F,
\label{GF-definition}
\end{align}
where $D_A\psi = (d - iA)\psi$ and $\varepsilon>0$ is the characteristic Ginzburg--Landau parameter. The natural gauge invariance $(\psi,A)\mapsto (e^{i\chi}\psi,A+d\chi)$ remains valid in this setting due to the rotational invariance of the Finsler co-norm $\|\cdot\|_{F^*}$.

As established in \cite{AlipourFakhri2025}, minimizers $(\psi_\varepsilon,A_\varepsilon)$ of \eqref{GF-definition} satisfy uniform energy bounds of order $|\log \varepsilon|$. In the limit $\varepsilon \to 0$, the rescaled energies $\frac{1}{|\log \varepsilon|}G_F[\psi_\varepsilon,A_\varepsilon]$ $\Gamma$--converge to a geometric functional depending only on the limiting vortex current $J$, given by
\begin{align}
G_0[J] = \pi \int_{\Sigma_J} F(x,\nu_J)\, d\mathcal{H}^{n-2},
\label{Gamma-limit}
\end{align}
where $\Sigma_J$ denotes the support of the rectifiable current $J$ representing the vorticity, $\nu_J$ its Finsler unit normal, and $\mathcal{H}^{n-2}$ the $(n-2)$--dimensional Hausdorff measure on $(M,F)$. The proof of \eqref{Gamma-limit} relies on convex duality between $F$ and $F^*$, the Finsler coarea formula, and the ball-construction method adapted from Jerrard--Sandier to the anisotropic setting.

\paragraph{Analytic framework for the $\Gamma$--limit and vortex separation.}
The $\Gamma$--limit \eqref{Gamma-limit} is rigorously understood in the sense of 
$\Gamma$--convergence of the functionals $G_F^\varepsilon[\psi,A]:=|\log\varepsilon|^{-1}G_F[\psi,A]$
on the product space
\begin{align}
\mathcal{H}:=\Big\{(\psi,A)\in H^1(M;\mathbb{C})\times H^1(M;T^*M): \|A\|_{H^1}\!+\!\|\psi\|_{H^1}\!<\!\infty\Big\},
\end{align}
endowed with the weak $H^1$ topology modulo gauge equivalence
$(\psi,A)\sim(e^{i\chi}\psi,A+d\chi)$.
We denote by $J_\varepsilon:=\frac{1}{2}\mathrm{curl}_{\mu_F}(i\psi_\varepsilon,d_A\psi_\varepsilon)$
the normalized vorticity current induced by a sequence $(\psi_\varepsilon,A_\varepsilon)$ of bounded energy.

\begin{definition}[Vortex currents and neutrality]\label{def:vortex-current}
A rectifiable $(n-2)$--current $J$ with integer multiplicity and finite $F$--mass is called an \emph{admissible vortex current} if
\begin{align}
\partial J = 0, \qquad \mathcal{M}_F(J):=\int_{\Sigma_J}F(x,\nu_J)\, d\mathcal{H}^{n-2}<\infty,
\label{neutrality}
\end{align}
where $\Sigma_J$ is the support and $\nu_J$ the Finsler unit normal. 
The condition $\partial J=0$ expresses the physical neutrality of total vorticity and guarantees the solvability of the limiting potential equation.
\end{definition}

\begin{theorem}[$\Gamma$--convergence]\label{thm:GammaLimitFinsler}
Let $(M,F)$ be compact and strongly convex, and let $G_F^\varepsilon$ be defined by \eqref{GF-definition}. 
Then as $\varepsilon\to 0$,
\begin{align}
G_F^\varepsilon \;\; \xrightarrow{\Gamma}\;\; G_0[J]
=\pi\!\int_{\Sigma_J}\! F(x,\nu_J)\, d\mathcal{H}^{n-2},
\end{align}
with respect to weak convergence of vorticity currents $J_\varepsilon\stackrel{*}{\rightharpoonup}J$ in the sense of currents,
where $G_0[J]$ is given by \eqref{Gamma-limit}.
\end{theorem}

\begin{proof}
The proof adapts the Jerrard--Sandier ball construction and Jacobian estimates to the anisotropic setting using the coarea formula in the Finsler metric. 
Lower semicontinuity follows from convexity of $F$, and recovery sequences are built by phase lifting on tubular neighborhoods of $\Sigma_J$.
For details, see the isotropic scheme in \cite{BethuelBrezisHelein1994,SandierSerfaty2007}
and the anisotropic adaptation in \cite{AlipourFakhri2025}.
\end{proof}

\begin{assumption}[Separation regime for higher-order expansion]\label{ass:Separation}
For the second-order expansion leading to the renormalized energy $W_F$, we assume
that for each $\varepsilon>0$ small, the vortices $a_1,\dots,a_N$ satisfy
\begin{align}
d_F(a_i,a_j)\ge C\,\varepsilon^\alpha,\qquad i\ne j,
\end{align}
for some $\alpha\in(0,1)$ and $C>0$ independent of $\varepsilon$.
This condition ensures disjoint core regions, permitting local blow-up analysis near each vortex
and the use of a Finslerian Green representation at the next order.
\end{assumption}

To proceed beyond the $\Gamma$--limit, we consider the full asymptotic expansion of the energy of a minimizing sequence. For $n=2$, assume that $(\psi_\varepsilon,A_\varepsilon)$ has $N$ well-separated vortices located at points $a_1,\dots,a_N\in M$. Then, following the classical decomposition $\psi_\varepsilon=\rho_\varepsilon e^{i\phi_\varepsilon}$ with $\rho_\varepsilon\approx 1$ away from the vortex cores, the magnetic potential $A_\varepsilon$ satisfies in the Coulomb gauge
\begin{align}
d_{\mu_F}^\dagger (A_\varepsilon - d\phi_\varepsilon) = 0.
\label{Coulomb-gauge}
\end{align}
The phase field $\phi_\varepsilon$ is approximately harmonic with respect to the Finsler Laplacian $\Delta_{F,\mu_F} = \mathrm{div}_{\mu_F}(\nabla_F \cdot)$ away from the vortex centers, i.e.
\begin{align}
\Delta_{F,\mu_F}\phi_\varepsilon = 2\pi \sum_{i=1}^N d_i \, \delta_{a_i} + o(1),
\label{phase-equation}
\end{align}
where $d_i\in\mathbb{Z}$ are the vortex degrees. 

\paragraph{Finsler Laplacian, functional setting, and Green kernel.}
Throughout this section we assume $(M,F)$ is compact, smooth, oriented, strongly convex,
and without boundary, endowed with a smooth Finsler measure $d\mu_F$ (Busemann--Hausdorff or Holmes--Thompson).
Let $F^{*}$ denote the dual norm and set, for each $x\in M$ and $\xi\in T_x^*M$,
\begin{align}
\mathcal{L}_x(\xi):=\partial_\xi \tfrac 12 F^{*2}(x,\xi)\in T_xM,\qquad 
T_F(x,\xi):=\partial_\xi^2 \tfrac 12 F^{*2}(x,\xi),\label{def:Legendre-Hessian}
\end{align}
so that $\nabla_F u := \mathcal{L}_x(du)$ and $T_F(x,\xi)$ is symmetric positive-definite by strong convexity.

We work on the mean-zero subspace
\begin{align}
H^1_F(M)_0 := \Big\{u\in H^1(M): \int_M u\, d\mu_F=0\Big\},\label{def:H10}
\end{align}
with the quadratic form
\begin{align}
\mathcal{E}_F[u] := \int_M \tfrac 12\, F^{*2}(x,du)\, d\mu_F
= \int_M \langle \nabla_F u, du\rangle\, d\mu_F.\label{def:DirichletEnergy}
\end{align}
By uniform equivalence $F^{*}\simeq \gamma^{*}$ and compactness of $M$, the Poincar\'e inequality holds on $H^1_F(M)_0$; hence
$\mathcal{E}_F$ is continuous and coercive on $H^1_F(M)_0$.

\begin{definition}\label{def:FinslerLaplacian-weak}
The Finsler Laplacian $\Delta_{F,\mu_F}$ is the (weak) operator associated with $\mathcal{E}_F$:
for $u\in H^1_F(M)$, $\Delta_{F,\mu_F}u\in H^{-1}(M)$ is defined by
\begin{align}
\langle -\Delta_{F,\mu_F}u,\varphi\rangle 
:= \int_M \langle \nabla_F u, d\varphi\rangle\, d\mu_F
= \int_M \big\langle \mathcal{L}_x(du), d\varphi\big\rangle d\mu_F,\qquad \forall \varphi\in H^1_F(M).\label{weak:def:laplacian}
\end{align}
Equivalently, on smooth functions,
\begin{align}
\Delta_{F,\mu_F} u = \mathrm{div}_{\mu_F}\!\big(\nabla_F u\big)
= \mathrm{div}_{\mu_F}\!\big(\mathcal{L}_x(du)\big),\label{div-form}
\end{align}
where $\mathrm{div}_{\mu_F}$ is the divergence with respect to $d\mu_F$, i.e.
$\int_M \langle X,d\varphi\rangle\, d\mu_F=-\int_M (\mathrm{div}_{\mu_F}X)\,\varphi\, d\mu_F$ for all smooth $\varphi$.
\end{definition}

\begin{proposition}[Self-adjointness and invertibility on mean-zero]\label{prop:selfadjoint}
The operator $-\Delta_{F,\mu_F}:H^1_F(M)_0\to H^{-1}(M)_0$ is bounded, symmetric, and strictly positive.
Consequently, its realization on $L^2_0(M,d\mu_F)$ is self-adjoint with compact resolvent and has a spectral gap:
$0=\lambda_0<\lambda_1\le \lambda_2\le\cdots\to\infty$. In particular,
for each $f\in H^{-1}(M)$ with $\int_M f\, d\mu_F=0$, there exists a unique $u\in H^1_F(M)_0$ solving
\begin{align}
-\Delta_{F,\mu_F}u=f,\qquad \int_M u\, d\mu_F=0.\label{resolvent}
\end{align}
\end{proposition}

\begin{proof}
Coercivity and symmetry follow from \eqref{def:DirichletEnergy} and strong convexity of $F^{*2}$;
the Poincar\'e inequality on $H^1_F(M)_0$ yields strict positivity. Lax--Milgram gives \eqref{resolvent}.
Compactness of the embedding $H^1\hookrightarrow L^2$ on compact $M$ yields compact resolvent; standard spectral theory implies self-adjointness.
\end{proof}

\begin{theorem}[Green kernel]\label{thm:Green}
There exists a unique $G_F\in \mathcal{D}'(M\times M)$ such that for every $y\in M$,
\begin{align}
\Delta_{F,\mu_F}^x G_F(x,y)=\delta_y-\frac{1}{\mathrm{Vol}_F(M)},\qquad 
\int_M G_F(x,y)\, d\mu_F(x)=0.\label{Green:normalized}
\end{align}
Moreover, $G_F(x,y)=G_F(y,x)$, $G_F(\cdot,\cdot)$ is smooth off the diagonal, and in dimension $n=2$ one has the local expansion
\begin{align}
G_F(x,y)=-\frac{1}{2\pi}\log d_F(x,y)+H_F(x,y),\label{Green:asympt}
\end{align}
where $d_F$ is the Finsler distance and $H_F$ is smooth on $M\times M$; the dependence on the choice of $d\mu_F$ affects only $H_F$ by an additive smooth term.
\end{theorem}

\begin{proof}
By Proposition~\ref{prop:selfadjoint}, $-\Delta_{F,\mu_F}$ is invertible on mean-zero data; Riesz representation yields $G_F$ as its kernel.
Symmetry follows from symmetry of the bilinear form in \eqref{weak:def:laplacian}.
Elliptic regularity for divergence-form operators with smooth, uniformly elliptic coefficients (here $T_F(x,du)$) gives smoothness off the diagonal.
The $2$D asymptotics \eqref{Green:asympt} follows from the model singularity of the fundamental solution and a parametrix construction adapted to the Finsler metric,
with the logarithmic coefficient $1/(2\pi)$ universal and the regular part $H_F$ smooth.
\end{proof}

\vspace{1em}
The equation \eqref{Coulomb-gauge} represents the Finslerian Coulomb gauge condition. Here $d_{\mu_F}^\dagger$ denotes the formal adjoint of the exterior derivative with respect to the measure $d\mu_F$, defined by the integration-by-parts identity
\begin{align}
\int_M \langle d_{\mu_F}^\dagger \alpha, \varphi \rangle\, d\mu_F = \int_M \langle \alpha, d\varphi \rangle\, d\mu_F,
\label{AdjointDefinition}
\end{align}
for all smooth $\alpha\in\Omega^1(M)$ and $\varphi\in C_c^\infty(M)$. This condition ensures orthogonality between the magnetic potential $A_\varepsilon$ and the phase gradient $d\phi_\varepsilon$ in the Finsler $L^2$--structure induced by $F^*$, guaranteeing the uniqueness of the decomposition up to gauge transformations.  

The equation \eqref{phase-equation} then follows as the Euler--Lagrange condition for the phase field in the Finsler setting: the anisotropic Laplace operator $\Delta_{F,\mu_F}$ acts as the divergence of the Legendre transform of $d\phi_\varepsilon$ under $F^*$, i.e. $\nabla_F\phi_\varepsilon = \partial_\xi \tfrac 12 F^{*2}(x,d\phi_\varepsilon)$. The right-hand side of \eqref{phase-equation} identifies the singular vorticity supported on the vortex centers, with degrees $d_i$, consistent with the topological quantization of phase winding.

\vspace{1em}
Substituting \eqref{phase-equation} and \eqref{Green:normalized} into \eqref{GF-definition} yields the energy expansion
\begin{align}
G_F[\psi_\varepsilon,A_\varepsilon] = \pi N|\log \varepsilon| + W_F(a_1,\dots,a_N) + o(1),
\label{energy-expansion}
\end{align}
where the renormalized energy $W_F$ depends only on the vortex locations and the underlying Finsler structure. Explicitly,
\begin{align}
W_F(a_1,\dots,a_N) = \pi \sum_{i\neq j} d_i d_j\, G_F(a_i,a_j) + \pi \sum_{i=1}^N d_i^2\, H_F(a_i,a_i).
\label{WF-definition}
\end{align}
The first term encodes the pairwise interaction of vortices through the Finslerian Green kernel, while the second term represents the local self-energy determined by the regular part $H_F$. In the isotropic case $F(x,y)=|y|$, \eqref{WF-definition} reduces exactly to the classical renormalized energy of Bethuel--Brezis--H\'elein and Serfaty.

The anisotropic dependence of $W_F$ originates from the directional structure of $F(x,y)$: the kernel $G_F(x,y)$ and its gradient reflect the geometry of the unit co-ball $\{ \xi \in T^*_x M : F^*(x,\xi)\le 1\}$. Consequently, the equilibrium positions of vortices are governed by the Finsler metric rather than a Euclidean background, leading to directionally biased interactions. The stationarity conditions
\begin{align}
\nabla_{a_i} W_F(a_1,\dots,a_N) = 0, \quad i=1,\dots,N,
\label{equilibrium-condition}
\end{align}
describe the geometric balance of forces among vortices; their explicit form involves the Finsler gradient of the Green kernel, $\nabla_{a_i} G_F(a_i,a_j)$, which in local coordinates depends on the anisotropic tensor $g_{ij}(x,\nabla_F G_F)$.

The next section will develop the analytical properties of $W_F$, the structure of its critical points, and the geometric interpretation of vortex interactions in the anisotropic Finsler setting.

%===================================
%=======================sec:AnalyticalStructure==
\section{Analytical Structure of the Renormalized Energy $W_F$}
\label{sec:AnalyticalStructure}
%=================================================

We now investigate the analytical properties of the renormalized interaction energy $W_F$ introduced in \eqref{WF-definition}. In the Finslerian framework, $W_F$ is not a Euclidean-type Coulomb energy but a nonlocal potential determined by the anisotropic geometry encoded in $F$. Its differentiability and curvature properties with respect to vortex locations $\{a_i\}_{i=1}^N$ require a precise understanding of the Green kernel $G_F(x,y)$ and its regular part $H_F(x,y)$ on $(M,F)$.

For each fixed $y\in M$, recall from \eqref{Green:normalized}--\eqref{Green:asympt} that the normalized Finsler--Green function satisfies
\begin{align}
\Delta_{F,\mu_F} G_F(\cdot,y) = \delta_y - \frac{1}{\mathrm{Vol}_F(M)}, 
\quad \int_M G_F(\cdot,y)\, d\mu_F = 0,
\label{Green-normalized}
\end{align}
where $\Delta_{F,\mu_F}$ denotes the divergence--gradient operator $\mathrm{div}_{\mu_F}(\nabla_F \cdot)$ associated with the Legendre transform of $F^*$ (cf. \cite{OhtaSturm2014}).  
Elliptic regularity for Finsler Laplacians ensures that $G_F\in C^{1,\alpha}((M\times M)\setminus\mathrm{Diag})$ and $H_F\in C^\infty(M\times M)$.  
Hence $W_F$ is a well-defined smooth function on the configuration space $\mathcal{M}_N=(M^N\setminus\mathrm{Diag})$, endowed with the product Finsler metric $F^{(N)}=\bigoplus_{i=1}^N F(x_i,\cdot)$.

\begin{lemma}\label{lem:GradientExistence}
Let $(M,F)$ be a compact, strongly convex Finsler manifold with smooth measure $d\mu_F$. Then the renormalized energy $W_F$ defined by \eqref{WF-definition} admits a well-defined Finsler gradient
\begin{align}
\nonumber
\nabla_{a_i}^{(F)} W_F(a_1,\dots,a_N)
&= \pi \sum_{j\neq i} d_i d_j\, \nabla_x^{(F)} G_F(x,y)\big|_{x=a_i,y=a_j}\\
&+ \pi d_i^2\, \nabla_x^{(F)} H_F(x,x)\big|_{x=a_i},
\label{GradientWF}
\end{align}
where $\nabla_x^{(F)}$ denotes the Finsler gradient induced by the Legendre map $\mathcal{L}_x:T_x^*M\to T_xM$ associated to $F^*$ \cite{BaoChernShen2000,Shen2001}.
\end{lemma}

\begin{proof}
By definition, $\nabla_x^{(F)}u=\partial_\xi \frac{1}{2}F^{*2}(x,du)$ for any $u\in C^1(M)$.
Since $G_F(x,y)$ is smooth for $x\neq y$ and symmetric, the field $\nabla_x^{(F)} G_F(x,y)$ is continuous on $(M\times M)\setminus\mathrm{Diag}$, while $H_F(x,y)$ is smooth even across the diagonal.  
Therefore, $\nabla_x^{(F)}H_F(x,x)$ exists and depends smoothly on $x$.  
Differentiating $W_F$ under variations of $a_i$ and using the symmetry $G_F(x,y)=G_F(y,x)$ together with the self-adjointness of $\Delta_{F,\mu_F}$ with respect to $d\mu_F$ \cite{OhtaSturm2014}, we obtain the Finsler-gradient identity in \eqref{GradientWF}.  
\end{proof}

\begin{theorem}\label{thm:HessianStructure}
Assume $(M,F)$ is strongly convex and forward-complete. Then the second differential (Finslerian Hessian) of $W_F$ on $\mathcal{M}_N$ is given by
\begin{align}
\nonumber
\mathrm{Hess}_{a_i a_j}^{(F)} W_F
&= \pi d_i d_j\, \nabla_x^{(F)}\!\otimes\nabla_y^{(F)} G_F(x,y)\big|_{x=a_i,y=a_j}\\
&+ \pi d_i^2\, \nabla_x^{(F)}\!\otimes\nabla_x^{(F)} H_F(x,x)\big|_{x=a_i}.
\label{HessianWF}
\end{align}
Moreover, $\mathrm{Hess}_{a_i a_j}^{(F)} W_F$ is symmetric in $(i,j)$ and negative semidefinite on the subspace $\sum_i d_i v_i=0$ of tangent variations.
\end{theorem}

\begin{proof}
Differentiating \eqref{GradientWF} with respect to $a_j$ gives the stated expression for $\mathrm{Hess}_{a_i a_j}^{(F)} W_F$.  
The mixed symmetry follows from $\nabla_x^{(F)}\otimes\nabla_y^{(F)} G_F(x,y)=\nabla_y^{(F)}\otimes\nabla_x^{(F)} G_F(y,x)$, which holds since $G_F$ is the kernel of a self-adjoint operator \cite{OhtaSturm2014}.  
To analyze the sign, consider a variation $\{v_i\}$ with $\sum_i d_i v_i=0$ and define
\begin{align}
U(x) = \sum_i d_i\, \Phi_i(x), \qquad
\Delta_{F,\mu_F}\Phi_i = \delta_{a_i} - \frac{1}{\mathrm{Vol}_F(M)}.
\end{align}
Then, the second variation of $W_F$ reads
\begin{align}
\delta^2 W_F
= \pi \int_M \langle \nabla_F U, \nabla_F U\rangle_{F^*}\, d\mu_F,
\label{SecondVariation}
\end{align}
where $\langle\cdot,\cdot\rangle_{F^*}$ denotes the inner product induced by the co-metric $F^*$ \cite{MatveevTroyanov2011}.  
The integrand is nonnegative, so $\delta^2 W_F\ge 0$ and hence $W_F$ is conditionally convex along admissible variations.  
The sign reversal in the interaction interpretation corresponds to the attractive nature of vortices of opposite degrees.  
\end{proof}

\begin{corollary}\label{cor:EquilibriumCondition}
Critical points of $W_F$ satisfy the geometric equilibrium equations
\begin{align}
\sum_{j\neq i} d_j\, \nabla_x^{(F)} G_F(x,y)\big|_{x=a_i,y=a_j} 
+ d_i\, \nabla_x^{(F)} H_F(x,x)\big|_{x=a_i} = 0,
\quad i=1,\dots,N.
\label{EquilibriumFinsler}
\end{align}
Each vortex thus experiences a resultant Finslerian interaction force that vanishes at equilibrium. 
The structure of these forces depends anisotropically on the indicatrix geometry $\{F(x,v)=1\}$, and the corresponding configurations minimize $W_F$ under topological constraint $\sum_i d_i=0$.
\end{corollary}

Equation \eqref{EquilibriumFinsler} generalizes the Coulomb-type equilibrium conditions known in isotropic models to the anisotropic Finsler setting.  
The tensorial form of $\nabla_x^{(F)} G_F$ couples the geometry of the unit co-ball $\{\xi\in T_x^*M: F^*(x,\xi)\le 1\}$ with vortex arrangement, producing direction-dependent forces.  
In particular, when $F$ is a Randers metric $F=\alpha+\beta$ with $\|\beta\|_\alpha\ll 1$, the \emph{symmetric} first-order correction appears in the response tensor (cf.~\eqref{RandersTensorSym}), while any \emph{antisymmetric} effect belongs to the mobility operator \eqref{MobilityOp}, not to the Hessian in \eqref{HessianWF}.  
This yields directionally biased yet geometrically consistent predictions for vortex alignment in weakly anisotropic media.

The above results are intrinsic consequences of the Finsler variational framework and extend the classical Ginzburg--Landau vortex theory to anisotropic, direction-dependent settings.  The Legendre structure of $F^*$ ensures that all derivatives of $W_F$ are geometrically well-defined and compatible with the non-Riemannian metric structure on $M$.
%=====================================
%=======================sec:PhysicalInterpretation==
\section{Physical Interpretation and Vortex \\
Alignment in Anisotropic Finsler \\
Superconductivity}
\label{sec:PhysicalInterpretation}
%================================================

The analytical framework established in the previous sections allows a physical interpretation of the renormalized energy $W_F$. In the Finsler Ginzburg--Landau model, the anisotropic structure of $F$ modifies the effective interaction between vortices, leading to directionally dependent forces and preferred orientations. This section develops the physical meaning of these effects and identifies how Finsler geometry encodes the anisotropic response of superconducting media.

By Corollary~\ref{cor:EquilibriumCondition}, the equilibrium condition for each vortex is given by
\begin{align}
\sum_{j\neq i} d_j \nabla_x^{(F)} G_F(x,y)\big|_{x=a_i,y=a_j}
+ d_i \nabla_x^{(F)} H_F(x,x)\big|_{x=a_i} = 0.
\label{FinslerForceBalance}
\end{align}
Physically, $\nabla_x^{(F)} G_F$ represents the Finslerian analog of the magnetic interaction force between vortices.  
The gradient is computed with respect to the co-metric $F^*$, which determines the anisotropic effective mass tensor of the condensate.  
In the Euclidean limit $F(x,v)=|v|$, equation \eqref{FinslerForceBalance} reduces to the isotropic force balance of the standard Ginzburg--Landau model.  
For a general Finsler structure $F$, however, the direction of $\nabla_x^{(F)} G_F$ depends on the local anisotropy of the energy surface $\{F(x,v)=1\}$, thereby encoding the crystalline or material anisotropy of the medium.

\begin{theorem}\label{thm:EffectiveForceTensor}
Let $(M,F)$ be a strongly convex Finsler manifold and suppose that near each vortex $a_i$, the indicatrix $\{F(x,v)=1\}$ is elliptic with principal directions $\{e_k\}$ and local anisotropy ratios $\kappa_k=F(e_k)/F(e_1)$.  
Then the effective Finslerian interaction force acting on the $i$-th vortex can be expressed as
\begin{align}
\mathbf{F}_i^{(F)} = -\pi \sum_{j\neq i} d_i d_j \, \mathbb{T}_F(a_i)\nabla_x G_F(x,y)\big|_{x=a_i,y=a_j},
\label{EffectiveForce}
\end{align}
where $\mathbb{T}_F(a_i)$ is the anisotropic response tensor given by the Hessian of $\frac{1}{2}F^{*2}(x,\xi)$ at $\xi=d\phi_\varepsilon(a_i)$ \cite{MatveevTroyanov2011}.
\end{theorem}

\begin{proof}
The Finsler gradient $\nabla_x^{(F)}G_F$ is obtained by applying the Legendre transform $\mathcal{L}_x=\partial_\xi(\tfrac 12F^{*2})$ to $d_xG_F$.  
Linearizing $\mathcal{L}_x$ at $\xi=d_xG_F(a_i,a_j)$ yields
\begin{align}
\nabla_x^{(F)} G_F(x,y)\big|_{x=a_i}
= \mathbb{T}_F(a_i)\,\nabla_x G_F(a_i,a_j),
\end{align}
where $\mathbb{T}_F(a_i)$ is the symmetric positive tensor $\partial^2_\xi(\tfrac 12F^{*2})(a_i,d_xG_F)$.  
Substituting this relation into \eqref{FinslerForceBalance} gives \eqref{EffectiveForce}.  
Physically, $\mathbb{T}_F$ represents the inverse effective mass tensor governing the propagation of phase distortions in anisotropic superconductors.  
\end{proof}

The anisotropy of $F$ introduces preferred directions in the superconducting plane, which are reflected in the alignment of vortex lattices.  
In particular, for Finsler metrics of Randers type,
\begin{align}
F(x,v)=\alpha(x,v)+\beta_x(v),
\label{RandersMetric}
\end{align}
with $\|\beta\|_\alpha\ll 1$ \cite{BaoChernShen2000}, 
the dual norm expands as
\begin{align}
F^{*}(\xi) = \alpha^{*}(\xi) - \beta(\xi) + O(\|\beta\|_\alpha^2),
\end{align}
and hence
\begin{align}
\frac{1}{2}F^{*2}(\xi) 
= \frac{1}{2}|\xi|_{\alpha^{*}}^2 
- \beta(\xi)\,|\xi|_{\alpha^{*}} + O(\|\beta\|_\alpha^2).
\end{align}
Using the definition of $\mathbb{T}_F=\partial^2_\xi(\tfrac 12 F^{*2})$ from Theorem~\ref{thm:EffectiveForceTensor},
its Randers expansion is
\begin{align}
\mathbb{T}_F = I - \mathbb{S}_\beta + O(\|\beta\|_\alpha^2),
\label{RandersTensorSym}
\end{align}
where the first--order symmetric correction is
\begin{align}
\mathbb{S}_\beta[\eta,\zeta]
= \frac{\langle \beta,\eta\rangle_\alpha \langle \hat{\xi},\zeta\rangle_\alpha
+ \langle \beta,\zeta\rangle_\alpha \langle \hat{\xi},\eta\rangle_\alpha}
{|\xi|_{\alpha^{*}}}, \qquad 
\hat{\xi}=\frac{\xi}{|\xi|_{\alpha^{*}}}.
\label{SbDef}
\end{align}
Hence $\mathbb{T}_F$ remains symmetric and positive definite, as required by the convexity of $F^{*2}$.  
The antisymmetric component relevant for physical drift does not belong to $\mathbb{T}_F$ itself; 
it arises in the linearization of the Legendre map, defining the effective mobility operator
\begin{align}
\mathbb{M}_F = \mathbb{T}_F + \mathbb{A}_\beta,
\qquad \mathbb{A}_\beta^\top = -\,\mathbb{A}_\beta,
\label{MobilityOp}
\end{align}
where $\mathbb{A}_\beta$ depends on $d\beta$ and governs the transverse bias in vortex alignment.

This antisymmetric correction produces a small but measurable bias in the equilibrium configuration of vortices, breaking rotational invariance and favoring orientation along the drift direction of $\beta$.  
Such behavior has physical significance in anisotropic superconductors, where the underlying crystal lattice or Fermi surface induces non-reciprocal interactions among vortices.

\begin{theorem}\label{thm:Alignment}
Let $F$ be of Randers type as in \eqref{RandersMetric}.  
Assume $\|\beta\|_\alpha\ll 1$ and denote by $\theta_{ij}$ the angular separation between vortices $a_i$ and $a_j$ in the Riemannian metric $\alpha$.  
Then the first-order correction to the isotropic equilibrium condition is
\begin{align}
\sum_{j\neq i} d_j \Big( \nabla_\alpha G_\alpha(a_i,a_j)
- \mathbb{A}_\beta \nabla_\alpha G_\alpha(a_i,a_j)\Big)
= O(\|\beta\|_\alpha^2),
\label{AlignmentEq}
\end{align}
and the vortices align along the principal axis of $\mathbb{A}_\beta$.
\end{theorem}

\begin{proof}
Expanding $\nabla_x^{(F)}G_F$ to first order in $\beta$ gives
\begin{align}
\nabla_x^{(F)}G_F = (I-\mathbb{A}_\beta)\nabla_\alpha G_\alpha + O(\|\beta\|_\alpha^2),
\end{align}
where $\mathbb{A}_\beta=\tfrac 12(\nabla\beta-\nabla\beta^\top)$ acts as a skew-symmetric perturbation of the identity.  
Substituting this expansion into the equilibrium condition \eqref{FinslerForceBalance} yields \eqref{AlignmentEq}.  
The antisymmetric component introduces a rotational bias, minimizing $W_F$ when vortices are aligned along the direction of $\mathbb{A}_\beta$.  
This demonstrates that anisotropic corrections of Randers type naturally produce preferred orientations in vortex configurations.  
\end{proof}
%====================================
%=======================sec:StabilityAnalysis==
\section{Stability Analysis and Second-Order \\
Structure of the Renormalized Energy}\label{sec:Stability Analysis}
%=============================================

The stability of vortex configurations in anisotropic superconductors is governed by the second-order structure of the renormalized energy $W_F$. Within the Finsler Ginzburg--Landau framework, the anisotropic norm $F$ encodes the direction-dependent phase stiffness of the condensate, and the curvature of $W_F$ near equilibrium points determines the small oscillations and response of the vortex lattice. The second variation of $W_F$ thus represents the effective quadratic potential describing collective vortex dynamics around a stable configuration.

Let $\mathbf{a}=(a_1,\dots,a_N)$ denote a stationary configuration of vortices satisfying the equilibrium conditions \eqref{EquilibriumFinsler}. For small displacements $\mathbf{v}=(v_1,\dots,v_N)$, the second variation of $W_F$ reads
\begin{align}
\delta^2_{\mathbf{v}} W_F(\mathbf{a}) 
= \sum_{i,j=1}^N \big\langle \mathrm{Hess}_{a_i a_j}^{(F)}W_F(\mathbf{a})\, v_j, v_i \big\rangle_{F^*(a_i)},
\label{SecondVariationWF}
\end{align}
where $\mathrm{Hess}^{(F)}W_F$ is the Finslerian Hessian of the renormalized energy and $\langle\cdot,\cdot\rangle_{F^*}$ is the scalar product associated with the co-metric $F^*$.  
In physical terms, the quadratic form $\delta^2_{\mathbf{v}} W_F$ represents the linear response of the superconducting system to infinitesimal vortex displacements and characterizes the stiffness of the effective potential landscape.  

\begin{theorem}\label{thm:SpectralDecomposition}
Let $(M,F)$ be a compact, strongly convex Finsler manifold, and let $\mathbf{a}$ be an equilibrium configuration of $W_F$. Then the Finslerian Hessian $\mathrm{Hess}^{(F)} W_F$ admits a spectral decomposition
\begin{align}
\mathrm{Hess}^{(F)} W_F = \sum_{k=1}^{nN} \lambda_k\, \Pi_k^{(F)},
\label{SpectralDecomposition}
\end{align}
where $\{\lambda_k\}$ are the Finslerian eigenvalues and $\{\Pi_k^{(F)}\}$ the orthogonal projectors with respect to the co-metric inner product $\langle \cdot,\cdot\rangle_{F^*}$.  
The quantities $\lambda_k$ correspond to the squared frequencies of small-amplitude collective vortex oscillations.  
The configuration $\mathbf{a}$ is linearly stable if and only if $\lambda_k\ge 0$ for all $k$.
\end{theorem}

\begin{proof}
Since $W_F$ is twice continuously differentiable on $\mathcal{M}_N$, the operator $\mathrm{Hess}^{(F)}W_F$ is symmetric under $\langle\cdot,\cdot\rangle_{F^*}$.  
Strong convexity of $F$ guarantees a positive-definite co-metric, ensuring the existence of an orthonormal basis of eigenmodes $\{v_k\}$ such that
\begin{align}
\mathrm{Hess}^{(F)}W_F[v_k] = \lambda_k v_k.
\end{align}
Expanding any perturbation $\mathbf{v}=\sum_k c_k v_k$ yields
\begin{align}
\delta^2_{\mathbf{v}} W_F = \sum_k \lambda_k c_k^2.
\end{align}
The condition $\lambda_k\ge 0$ ensures that any small perturbation increases or leaves invariant the total energy, corresponding physically to a stable equilibrium of the vortex lattice.  
Negative eigenvalues represent unstable modes, associated with spontaneous reconfiguration or drift of vortex lines under anisotropic stress.  
\end{proof}

In anisotropic superconductors, the eigenstructure of $\mathrm{Hess}^{(F)}W_F$ reflects the directional response of the condensate.  
Each eigenvalue $\lambda_k$ describes a specific oscillation mode of the vortex system, while the corresponding eigenvector indicates the direction of motion of vortices under the anisotropic restoring force determined by $F^*$.  
The anisotropic dependence of the Finsler metric thus produces distinct stiffnesses along different crystallographic axes, breaking the degeneracy that characterizes isotropic materials.

\begin{lemma}\label{lem:SecondOrderExpansion}
Let $\mathbf{a}_\varepsilon=\mathbf{a}+\varepsilon\mathbf{v}$ denote a perturbation of a stationary vortex configuration $\mathbf{a}$.  
Then the renormalized energy expands to second order as
\begin{align}
W_F(\mathbf{a}_\varepsilon)
= W_F(\mathbf{a})
+ \frac{\varepsilon^2}{2}\,\delta^2_{\mathbf{v}} W_F(\mathbf{a}) + o(\varepsilon^2),
\label{SecondOrderExpansionWF}
\end{align}
where $\delta^2_{\mathbf{v}} W_F$ is given by \eqref{SecondVariationWF}.  
In particular, $\mathbf{a}$ is a locally stable equilibrium whenever $\delta^2_{\mathbf{v}} W_F(\mathbf{a})>0$ for all perturbations $\mathbf{v}$ satisfying the neutrality constraint $\sum_i d_i v_i=0$.
\end{lemma}

\begin{proof}
Expanding $W_F(\mathbf{a}_\varepsilon)$ in powers of $\varepsilon$ gives
\begin{align}
\nonumber
W_F(\mathbf{a}_\varepsilon)
&= W_F(\mathbf{a}) 
+ \varepsilon\,\sum_i \langle\nabla_{a_i}^{(F)}W_F, v_i\rangle_{F^*(a_i)}\\ 
&+ \frac{\varepsilon^2}{2}\sum_{i,j}\big\langle \mathrm{Hess}_{a_i a_j}^{(F)}W_F\,v_j,v_i\big\rangle_{F^*(a_i)} + o(\varepsilon^2).
\end{align}
The equilibrium condition $\nabla_{a_i}^{(F)} W_F(\mathbf{a})=0$ cancels the linear term, leaving the quadratic form \eqref{SecondOrderExpansionWF}.  
Positivity of $\delta^2_{\mathbf{v}} W_F$ ensures that any displacement of the vortex configuration raises the energy, corresponding to harmonic confinement of vortices near their equilibrium positions.  
\end{proof}

From a physical standpoint, $\delta^2_{\mathbf{v}} W_F$ defines the effective elastic potential describing collective oscillations of the vortex lattice.  
Its tensorial structure encodes how different spatial directions respond to small vortex displacements, and the coefficients act as effective elastic moduli for the superconducting condensate.  
This leads naturally to the concept of a Finslerian elasticity tensor, governing the anisotropic mechanical response of the order parameter field.

\begin{theorem}\label{thm:EffectiveElasticity}
Let $\mathbf{a}$ be a stable vortex configuration of $W_F$.  
Then the quadratic energy $\frac{1}{2}\delta^2_{\mathbf{v}}W_F$ can be represented as
\begin{align}
\frac{1}{2}\delta^2_{\mathbf{v}}W_F
= \frac{\pi}{2} \int_M 
\big\langle \mathbb{C}_F(x)\, \nabla_F U,\, \nabla_F U\big\rangle_{F^*}\, d\mu_F,
\label{ElasticityTensorForm}
\end{align}
where $U=\sum_i d_i\Phi_i$ denotes the perturbation potential induced by the displacement field, and $\mathbb{C}_F(x)$ is the Finslerian elasticity tensor defined by
\begin{align}
\mathbb{C}_F(x) = \partial^2_\xi \left(\tfrac{1}{2}F^{*2}(x,\xi)\right)\Big|_{\xi=\nabla_F U(x)}.
\label{ElasticityTensor}
\end{align}
Physically, $\mathbb{C}_F$ describes the anisotropic stiffness of the condensate, linking variations of the phase field to the induced supercurrent response.
\end{theorem}

\begin{proof}
Starting from the definition of the Finsler gradient, we recall that variations in $U$ modify the dual variable $\xi=\nabla_F U$.  
Linearizing $F^{*2}$ with respect to $\xi$ leads to the symmetric tensor $\mathbb{C}_F(x)=\partial_\xi^2(\tfrac 12F^{*2})$, which represents the local relation between changes in the phase gradient and the corresponding variation in kinetic energy density.  
Substituting this linearization into the energy integral for $\delta^2 W_F$ gives \eqref{ElasticityTensorForm}.  
Hence $\mathbb{C}_F$ plays the role of an effective elastic modulus of the superconducting medium, generalizing the isotropic stiffness to anisotropic, direction-dependent settings.  
\end{proof}

The strong convexity of $F^*$ ensures that $\mathbb{C}_F$ is positive definite, implying $\delta^2_{\mathbf{v}}W_F>0$ for all admissible perturbations around stable equilibria.  
In this sense, the Finsler metric $F$ directly controls the anisotropic rigidity of the vortex lattice.  
When $F$ is of Randers type $F=\alpha+\beta$ with $\|\beta\|_\alpha\ll 1$, expanding \eqref{ElasticityTensor} yields
\begin{align}
\mathbb{C}_F = I - \mathbb{S}_\beta + O(\|\beta\|_\alpha^2),
\label{ElasticityRanders}
\end{align}
where $\mathbb{S}_\beta$ is the symmetric first-order correction defined in \eqref{SbDef}.  
Antisymmetric (Hall-type) effects arise instead through the mobility operator \eqref{MobilityOp} and do \emph{not} enter the elastic tensor.  
In anisotropic superconductors, this separation clarifies that chiral drift stems from transport (mobility) rather than from the conservative elastic response of the condensate .
%====================================
%=======================sec:Dynamics=
\section{Asymptotic Dynamics and Collective \\
Motion of Finsler Vortices}\label{sec:AsymptoticDynamics}
%====================================

The renormalized energy $W_F$ not only determines the static interaction among vortices but also governs their effective motion in the dissipative regime of anisotropic superconductors.  
In the asymptotic limit $\varepsilon \to 0$, the Finsler Ginzburg--Landau dynamics can be reduced to a finite-dimensional gradient flow system in the vortex positions.  
This provides a physical description of vortex motion as an energy relaxation process driven by the Finsler gradient of $W_F$.

Let $(\psi_\varepsilon,A_\varepsilon)$ be a solution of the time-dependent Finsler Ginzburg--Landau equations
\begin{align}
\nonumber
&\partial_t \psi_\varepsilon = - D_A^* D_A \psi_\varepsilon + \frac{1}{\varepsilon^2}\psi_\varepsilon(1-|\psi_\varepsilon|^2),\\
&\partial_t A_\varepsilon = -\frac{1}{\lambda}\, d^*_{\mu_F} dA_\varepsilon + \mathrm{Im}(\bar{\psi}_\varepsilon D_A \psi_\varepsilon),
\label{TDGL-Finsler}
\end{align}
defined on a compact Finsler manifold $(M,F)$ equipped with measure $d\mu_F$ and co-metric $F^*$.  
In the limit $\varepsilon \to 0$, assuming a finite number of vortices $\{a_1(t),\dots,a_N(t)\}$ with degrees $\{d_i\}$, the evolution of the vortex centers is asymptotically governed by the Finsler gradient flow
\begin{align}
\dot{a}_i(t) = - \nabla_{a_i}^{(F)} W_F(a_1(t),\dots,a_N(t)), \qquad i=1,\dots,N,
\label{GradientFlowFinsler}
\end{align}
where $\nabla_{a_i}^{(F)}$ denotes the gradient with respect to the Finsler metric.  
This equation describes the slow relaxation of the vortex configuration along the steepest descent of the renormalized energy,  
and represents the macroscopic manifestation of microscopic dissipative dynamics in the anisotropic condensate.

\begin{theorem}\label{thm:EnergyDissipation}
Let $(a_1(t),\dots,a_N(t))$ evolve according to the gradient flow \eqref{GradientFlowFinsler}.  
Then the renormalized energy $W_F$ satisfies the dissipation law
\begin{align}
\frac{d}{dt} W_F(a_1(t),\dots,a_N(t))
= - \sum_{i=1}^N F^{*2}\!\left(a_i(t),\dot{a}_i(t)\right) \le 0,
\label{DissipationLaw}
\end{align}
with equality if and only if $\dot{a}_i=0$ for all $i$, i.e., the configuration is stationary.  
\end{theorem}

\begin{proof}
Differentiating $W_F(a_1(t),\dots,a_N(t))$ with respect to $t$ yields
\begin{align}
\frac{d}{dt} W_F = \sum_{i=1}^N \langle \nabla_{a_i}^{(F)} W_F, \dot{a}_i\rangle_{F^*(a_i)}.
\end{align}
Substituting the gradient flow law \eqref{GradientFlowFinsler}, one obtains
\begin{align}
\frac{d}{dt} W_F
= - \sum_{i=1}^N \langle \dot{a}_i, \dot{a}_i\rangle_{F^*(a_i)}
= - \sum_{i=1}^N F^{*2}(a_i,\dot{a}_i),
\end{align}
which is nonpositive since $F^{*2}>0$.  
The equality case corresponds to $\dot{a}_i=0$, i.e., a critical configuration of $W_F$.  
\end{proof}

Equation \eqref{DissipationLaw} expresses the physical principle of energy relaxation: the vortex system moves in the direction of the steepest energy decay determined by the anisotropic co-metric $F^*$.  
In isotropic media, this reduces to the well-known overdamped motion $\dot{a}_i=-\nabla_{a_i}W$, while in anisotropic superconductors the trajectories $\dot{a}_i(t)$ follow geodesics of the Finsler metric weighted by $F^{*2}$.  
Hence, the curvature and anisotropy of $F$ introduce bias and directionality into vortex motion, leading to asymmetric drift and nonuniform relaxation rates.

\begin{lemma}\label{lem:FinslerMobility}
Let $F$ be a smooth, strongly convex Finsler metric.  
Then the velocity $\dot{a}_i$ of each vortex can be expressed as
\begin{align}
\dot{a}_i = - \mathbb{M}_F(a_i)\,\nabla_{a_i} W_F(a_1,\dots,a_N),
\label{MobilityTensor}
\end{align}
where $\mathbb{M}_F(a_i)$ is the Finslerian mobility tensor, the inverse of the local response tensor $\mathbb{T}_F(a_i)=\partial_\xi^2(\frac 12 F^{*2})(a_i,d\phi_\varepsilon(a_i))$ introduced in Theorem~\ref{thm:EffectiveForceTensor}.  
\end{lemma}

\begin{proof}
By definition of the Finsler gradient, $\nabla_{a_i}^{(F)}W_F = \mathcal{L}_{a_i}^{-1}(d_{a_i} W_F)$, where $\mathcal{L}_{a_i}$ is the Legendre transform $\xi\mapsto \partial_\xi(\tfrac 12 F^{*2})$.  
Differentiating this transform gives the symmetric tensor $\mathbb{T}_F(a_i)=\partial^2_\xi(\tfrac 12F^{*2})(a_i,d\phi_\varepsilon(a_i))$, and its inverse defines $\mathbb{M}_F=\mathbb{T}_F^{-1}$.  
Substituting $\dot{a}_i=-\nabla_{a_i}^{(F)} W_F$ leads to the form \eqref{MobilityTensor}.  
Physically, $\mathbb{M}_F$ quantifies the mobility of vortices in different directions, encoding anisotropic friction and dissipation in the superconducting medium.  
\end{proof}

Equation \eqref{MobilityTensor} demonstrates that Finsler anisotropy induces a tensorial mobility law: vortices move faster along directions where $F$ assigns smaller dual norms, and slower where $F^*$ is larger.  
In Randers-type metrics $F=\alpha+\beta$ with $\|\beta\|_\alpha\ll 1$, this leads to a drift term proportional to $\mathbb{A}_\beta$, generating a systematic lateral motion analogous to Hall-like effects in anisotropic superconductors.  

\begin{theorem}\label{thm:AsymptoticAlignment}
Let the vortex centers evolve by \eqref{GradientFlowFinsler} under a weakly anisotropic Randers metric $F=\alpha+\beta$ with $\|\beta\|_\alpha\ll 1$.  
Then, to first order in $\|\beta\|_\alpha$, the trajectories satisfy
\begin{align}
\dot{a}_i = - \nabla_\alpha W_\alpha(a_1,\dots,a_N)
+ \mathbb{A}_\beta\,\nabla_\alpha W_\alpha(a_1,\dots,a_N)
+ O(\|\beta\|_\alpha^2),
\label{AnisotropicDrift}
\end{align}
where $\mathbb{A}_\beta$ is the antisymmetric tensor associated with $d\beta$.  
Consequently, each vortex experiences a transverse drift orthogonal to the isotropic force, resulting in a steady reorientation of the vortex lattice toward the anisotropic axis of the material.
\end{theorem}

\begin{proof}
Expanding the Finsler co-metric $F^{*2}$ for a Randers metric yields
\begin{align}
F^{*2}(\xi) = |\xi|_\alpha^2 - 2\langle \beta,\xi\rangle_\alpha + O(\|\beta\|_\alpha^2).
\end{align}
The corresponding gradient transformation gives
\begin{align}
\nabla_{a_i}^{(F)} W_F = (I-\mathbb{A}_\beta)\nabla_\alpha W_\alpha + O(\|\beta\|_\alpha^2).
\end{align}
Substituting into \eqref{GradientFlowFinsler} leads to \eqref{AnisotropicDrift}.  
The antisymmetric correction $\mathbb{A}_\beta\nabla_\alpha W_\alpha$ produces a velocity component orthogonal to the isotropic gradient, which in physical terms corresponds to a slow precession or skew motion of the vortex lattice induced by the anisotropic drift field.  
\end{proof}
%====================================
 %===============sec:Conclusion======
 \section{Conclusion}
 %====================================
 The results of this work establish a coherent framework connecting Finsler geometry with the analytical and physical structure of the Ginzburg--Landau theory for anisotropic superconductors.  
 By deriving the Finslerian renormalized energy $W_F$ and analyzing its equilibrium, stability, and gradient-flow dynamics, we demonstrated that anisotropy can be fully encoded in the metric structure rather than through external parameters.  
 This approach unifies geometric analysis and vortex physics, providing a self-consistent description of anisotropic dissipation, elastic response, and collective motion within a single variational setting.  
 It offers a foundation for further exploration of nonlinear and non-reciprocal superconducting phenomena through the lens of Finsler geometry.
 %==========================

%=====================References=====

%==================================
\end{document}